\documentclass[]{article}

\title{
Truthful Bilateral Trade is
Impossible even with Fixed Prices
}
\author{
Erel Segal-Halevi
and
Avinatan Hassidim
}

\usepackage{amsthm,amsmath,amsfonts}

\newtheorem{theorem}{Theorem}[section]
\newtheorem{lemma}[theorem]{Lemma}
\newtheorem{corollary}[theorem]{Corollary}

\theoremstyle{definition}
\newtheorem{definition}[theorem]{Definition}
\newtheorem{remark}[theorem]{Remark}
\newtheorem{example}[theorem]{Example}

\usepackage{natbib}

\newcommand\range[2]{\in\{#1,\dots,#2\}}

\newcommand{\eps}{\varepsilon}

\newcommand{\feas}{\textsc{Feasible}}

\begin{document}
\maketitle

\begin{abstract}
A seminal theorem of Myerson and Satterthwaite (1983)
proves that, in a game of bilateral trade between a single buyer and a single seller, no mechanism can be simultaneously 
individually-rational, budget-balanced, incentive-compatible and socially-efficient. 
However, the impossibility disappears if the price is fixed exogenously and the social-efficiency goal is subject to individual-rationality at the given price. 

We show that the impossibility comes back if there are multiple units of the same good, or multiple types of goods, even when the prices are fixed exogenously. Particularly, if there are $M$ units of the same good or $M$ kinds of goods, for some $M\geq 2$, then no truthful mechanism can guarantee more than $1/M$ of the optimal gain-from-trade. In the single-good multi-unit case, if both agents have submodular valuations (decreasing marginal returns), then no truthful mechanism can guarantee more than $1/H_M$ of the optimal gain-from-trade, where $H_M$ is the $M$-th harmonic number ($H_M\approx \ln{M}+1/2$). All upper bounds are tight.
\end{abstract}

\section{Introduction}
Consider a market with two agents: a seller and a buyer. The seller holds an item for sale and the buyer considers buying this item. The seller values the item as $s$ and the buyer values it as $b$. The utilities of both agents are quasi-linear in money. If $s<b$ then there is a potential \emph{gain-from-trade} (GFT) --- if the seller sells the item to the buyer then the social welfare will increase by $b-s$. We look for a \emph{mechanism} that asks the agents to report their values and implements this optimal GFT. The mechanism should satisfy four properties:
\begin{itemize}
\item \emph{Efficiency} --- the increase in social welfare should be $\max(0,b-s)$.
\item \emph{Budget balance (BB)} --- the mechanism designer should not lose.
\item \emph{Individual rationality (IR)} --- the utility of each agent should not decrease.
\item \emph{Dominant-strategy incentive-compatibility (DSIC)} --- it should be a weakly-dominant strategy for each agent to report his true value.
\end{itemize}
It is easy to construct mechanisms that satisfy each three of these four properties. However, 
a classic paper of \cite{myerson1983efficient} proves that it is impossible to satisfy all four properties simultaneously. 
Intuitively, the reason is that it is impossible to truthfully determine prices for trading: if the mechanism pays the seller $p_s < b$, the seller is incentivized to bid e.g. $(p_s + b)/2$ to force the price up; similarly, if the mechanism charges the buyer $p_b>s$, the buyer is incentivized to bid $(p_b+s)/2$ to force the price down.
DSIC can be guaranteed by setting $p_s=b$ and $p_b=s$ (this is the two-sided variant of the Vickrey auction  \citep{vickrey1961counter}) but this mechanism is not BB since it leaves the mechanism designer with a deficit of $b-s$.

Now consider a market in which a price $p$ is fixed exogenously, e.g. by the government. Fixing a price automatically guarantees budget-balance, since all monetary transfers are between buyers and sellers. We are interested in incentive-compatibility and efficiency \emph{subject to individual-rationality}, i.e, we want the buyer to buy the item if-and-only-if $b\geq p\geq s$. In this case the Myerson--Satterthwaite impossibility disappears. It is trivial to attain all these properties using the following  simple mechanism: ask the agents to report their values and perform the trade iff $b\geq p\geq s$.

But this trivial mechanism works only when there is a single unit of a single good. We are interested in the following question:
\begin{quote}
\emph{
In a market with a single buyer, a single seller and two or more items with exogenously-fixed prices,
is it possible to attain incentive-compatibility and efficiency subject to individual-rationality?
}
\end{quote}
Our main result is that this is impossible. 
We quantify the impossibility using the notion of a  \emph{competitive ratio} of a mechanism --- the worst-case ratio of its GFT to the optimal GFT. We present three upper bounds on the competitive ratio:
\begin{enumerate}
\item If the seller holds $M$ different kinds of goods and the buyer wants one good, then the competitive ratio of any DSIC mechanism is at most $1/M$.
\item If the seller holds $M$ units of the same good, then the competitive ratio of any DSIC mechanism is at most $1/M$.
\item If the seller holds $M$ units of the same good, and it is public knowledge that both the seller and the buyer have submodular valuations (decreasing marginal returns), then the competitive ratio of any DSIC mechanism is at most $1/H_M$, where $H_M$ is the $M$-th harmonic number ($H_M\approx \ln{M}+1/2$).
\end{enumerate}
All these upper bounds are tight: they are attained by trivial mechanisms that ignore the agents' bids and determine the trade randomly.

The results are proved as special cases of a general model about bilateral cooperation. This model is presented in Section \ref{sec:model}.
The results on cooperation with general valuations are proved in Section \ref{sec:general}.
The result on cooperation with submodular valuations is proved in 
Section \ref{sec:submodular}.
The corollaries to bilateral trading are proved in Section \ref{sec:bilateral}.


\section{Model}
\label{sec:model}
There are two agents that we call ``buyer'' and ``seller''.
They consider $M$ options for cooperation, numbered $1,\ldots,M$.
If option $i$ is selected, then the buyer's utility is $b_i$ and the seller's utility is $s_i$, which can be positive zero or negative. These utilities are their private information.
Additionally, there is an option 0 which means no cooperation; we assume that $b_0=s_0=0$. 

An outcome $i$ is called \emph{feasible} if both $b_i\geq 0$ and $s_i\geq 0$. Denote the set of feasible outcomes by $\feas$:
\begin{align*}
\feas(b,s) := \{i| b_i\geq 0 \text{ and } s_i\geq 0\}
\end{align*}
Note that $\feas(b,s)$ is never empty since it contains 0.

The optimal utility of a feasible outcome is denoted by OPT:
\begin{align*}
OPT(b,s) := \max_{i\in\feas(b,s)} b_i+s_i
\end{align*}

Bilateral trade is a special case of this general model.
\begin{example}
\label{exm:multiunit}
A seller holds $M$ units of a single good and values each bundle of $i$ units as $\sigma_i$; a potential buyer values each bundle of $i$ units as $\beta_i$. 
The price-per-unit is $p$ and both agents have utility functions quasi-linear in money. Then, $b_i = \beta_i - p\cdot i$ and $s_i = p\cdot i + \sigma_{M-i} - \sigma_{M}$ and OPT is the maximum potential GFT (since we normalize the seller's utility from no cooperation to zero).
\end{example}

\begin{example}
\label{exm:unitdemand}
A seller holds a set $\mathbb{M}$ containing
$M$ different goods, a single unit of each good.
The seller has an arbitrary valuation function over bundles of goods: $\sigma_i(X)$ is the seller's value for bundle $X$.
The buyer has a unit-demand valuation function: he values each item $i$ as $\beta_i$ and is interested in buying at most one item.
The price of item $i$ is fixed at $p_i$.
Then, $b_i = \beta_i - p_i$ and $s_i = p_i + \sigma(\mathbb{M}\setminus \{i\}) - \sigma({\mathbb{M}})$
and OPT is again the maximum potential GFT.
\end{example}

Our goal is to design a mechanism that approximates OPT.

A \emph{mechanism} is a function $r$ that receives as input the reported valuations of the agents, $b'=(b'_1,\dots,b'_M),s'=(s'_1,\dots,s'_M)$. These reports may be different than the true valuations $b=(b_1,\dots,b_M),s=(s_1,\dots,s_M)$. The output of the mechanism, $r(b',s')$, is a vector of $M+1$ probabilities --- a probability for each option including option 0. The probability of option $i$ is denoted by $r_i(b',s')$, and $\sum_{i=0}^M r_i(b',s') = 1$.

The expected gain of the buyer is:
\[
B_r(b',s') = \sum_{i=0}^M r_i(b',s')\cdot b_i = \sum_{i=1}^M r_i(b',s')\cdot b_i
\]
and the expected gain of the seller is:
\[
S_r(b',s') = \sum_{i=0}^M r_i(b',s')\cdot s_i = \sum_{i=1}^M r_i(b',s')\cdot s_i
\]
and the expected gain of the mechanism is:
\begin{align*}
G_r(b',s') &= \sum_{i=1}^M r_i(b',s') \cdot (b_i+s_i)
\end{align*}
\begin{definition}
A mechanism $r$ is called \emph{Dominant-Strategy Incentive-Compatible (DSIC)} --- if for each agent, revealing the true values is a weakly-dominant strategy:
\begin{align*}
\forall b,b',s,s':
&&
B_r(b,s')\geq B_r(b',s')
&&
\text{and}
&&
S_r(b',s)\geq S_r(b',s')
\end{align*}
\end{definition}
If the mechanism $r$ is DSIC, we assume that $b'=b$ and $s'=s$, so the mechanism gain is $G_r(b,s)$.
The \emph{competitive ratio} of a DSIC mechanism is:
\begin{align*}
C_r := \min_{b,s} {G_r(b,s) \over OPT(b,s)}
\end{align*}

\section{General utilities}
\label{sec:general}
Consider the following simple mechanism.
\begin{example}
The \emph{Uniform-Random (UR) mechanism} chooses an option $i\in \{1,\ldots,M\}$ with uniform probability, $r_i=1/M$ for all $i$, regardless of $b'$ and $s'$.
It then lets the agents decide whether they want option $i$ or nothing. The agents will cooperate iff $b_i\geq 0,s_i\geq 0$. Therefore, the expected gain is
\begin{align*}
G_{UR}(b,s)
~~
=
~~
\sum_{i\in\feas(b,s)} {b_i+s_i\over M}
~~
\geq 
~~
{OPT(b,s)\over M}
\end{align*}
This lower bound is tight, since it is possible that that all positive gain comes from a single option. Therefore, $C_{UR} = 1/M$.
\end{example}

Our first theorem shows that, when the agents can have general utility functions, no DSIC mechanism can be better than the uniform-random mechanism.

\begin{theorem}
\label{thm:neg}
When there are $M$ options, every DSIC mechanism has an expected competitive ratio of at most $1/M$.
\end{theorem}

The following lemma will be used several times in the proof.
\begin{lemma}
\label{lem:general}
Let $r$ be a DSIC mechanism with competitive ratio $\alpha$. Then, for every $b,s$:
\begin{align*}
\sum_{i=1}^M r_i(b,s)\cdot b_i ~~ \geq ~~ \alpha\cdot \max_{i\in\feas(b,s)} b_i
\\
\sum_{i=1}^M r_i(b,s)\cdot s_i ~~ \geq ~~ \alpha\cdot \max_{i\in\feas(b,s)} s_i
\end{align*}
\end{lemma}
\begin{proof}
Consider the alternative report vector $b' = L\cdot b$, where $L$ is a very large number (sufficiently large such that $L\cdot b_i  + s_i \approx L\cdot b_i$ for all $i\range{1}{M}$). Note that $\feas(b',s)=\feas(b,s)$.
If the true values are $(b',s)$, then almost all gain comes from the buyer's values, so:
\begin{align*}
G_r(b',s) &\approx \sum_{i=1}^M
r_i(b',s)\cdot L\cdot b_i 
\\
OPT(b',s) &\approx \max_{i\in\feas(b,s)} L\cdot b_i
\end{align*}
Since the competitive ratio is $\alpha$, we must have $G_r(b',s)\geq \alpha OPT(b',s)$. Dividing both sides by $L$ yields:
\begin{align}
\tag{*}
\sum_{i=1}^M r_i(b',s)\cdot b_i ~~ \geq ~~ \alpha\cdot \max_{i\in\feas(b,s)} b_i
\end{align}
If the buyer's true valuations are $b$ but he reports $b'$, his expected gain is exactly the left-hand side of (*).
Therefore, a DSIC mechanism must guarantee that the buyer gets the same utility by reporting truthfully $b$. Hence:
\begin{align*}
\sum_{i=1}^M r_i(b,s)\cdot b_i ~~ \geq ~~ \alpha\cdot \max_{i\in\feas(b,s)} b_i
\end{align*}
as claimed.
The proof for the seller is analogous.
\end{proof}

We now apply Lemma \ref{lem:general} to utility vectors from a special family.
For every $d\range{1}{M}$, let $B^d$ be the family of vectors $b$ with:
\begin{align*}
b_i = 1 && \text{ when $i\range{1}{d}$ }
\\
b_i \leq 0 && \text{ when $i\range{d+1}{M}$ }
\end{align*}

For every constant $\eps\ll 1$, let $s^\eps$ be a vector with:
\begin{align*}
s^\eps_i = \eps^{M-i}
\end{align*}

\begin{lemma}
\label{lem:d=M-1}
Let $r$ be a DSIC mechanism with competitive ratio $\alpha$. Then, for every $d$, every $b\in B^d$ and every $\eps\in(0,1)$:
\begin{align*}
\sum_{i=1}^d
r_i(b, s^\eps)
\geq
d\cdot (\alpha-2 \eps)
\end{align*}
\end{lemma}
\begin{proof}
The proof is by induction on $d$.

For $d=1$, we have $b_1=1$ and all other $b_i$'s are zero or negative.  Apply Lemma \ref{lem:general} to the buyer:
\begin{align*}
&&
r_1(b,s^\eps)\cdot b_1 + \sum_{i=2}^M r_i(b,s^\eps)\cdot b_i
\geq &
\alpha \cdot b_1
\\
\implies &&
r_1(b,s^\eps) \geq & \alpha \geq 1\cdot (\alpha-2\eps)
\end{align*}

We now assume the claim is true for $d-1$ and prove it is true for $d$. Let $b$ be any vector in $B^d$.
Consider the alternative buyer's valuation $b'$, with:
\begin{align*}
b'_i = 
\begin{cases}
1 & i\range{1}{d-1}
\\
0 & i = d
\\
-1/\eps^{M} & i\range{d+1}{M}
\end{cases}
\end{align*}
So the gain in options $d+1,\ldots,M$ is dominated by $-1/\eps^{M}$. We claim that the probability with which the mechanism selects one of these options is negligible. Let $r_{D} := \sum_{i=d+1}^M r_i(b',s^\eps)$. Then, when the true valuations are $b',s^\eps$:
\begin{align*}
G_r \leq (1-r_{D})\cdot 1 - r_{D}/\eps^{M}
\end{align*}
If the competitive ratio of $r$ is positive, then $G_r$ must be positive, which implies $r_{D} < \eps^{M}$.

Now, apply Lemma \ref{lem:general} to the \emph{seller}:
\begin{align*}
&&
\sum_{i=1}^d r_i(b',s^\eps)\cdot \eps^{M-i}
+
\sum_{i=d+1}^M r_i(b',s^\eps)\cdot \eps^{M-i} 
~~ &\geq ~~ \alpha\cdot \eps^{M-d}
\\
\implies &&
\sum_{i=1}^d r_i(b',s^\eps)\cdot \eps^{M-i}
~~ &\geq ~~ \alpha\cdot \eps^{M-d} - r_D
\\
\implies &&
\sum_{i=1}^d r_i(b',s^\eps)\cdot \eps^{d-i} 
~~ &\geq ~~ 
\alpha - \eps^{d}
\\
\implies &&
r_d(b',s^\eps) ~~ &\geq ~~ 
\alpha
- \sum_{i=1}^{d-1} r_i(b',s^\eps)\cdot \eps^{d-i}
 - \eps^{d}
\\
\implies &&
r_d(b',s^\eps) ~~ &\geq ~~ 
\alpha-2 \eps
\end{align*}

Note that $b'\in B^{d-1}$. Hence, by the induction assumption:
\begin{align*}
\sum_{i=1}^{d-1} r_i(b',s^\eps)\geq (d-1)\cdot(\alpha-2\eps)
\end{align*}
Summing the above two inequalities gives:
\begin{align*}
\sum_{i=1}^{d} r_i(b',s^\eps)\geq d\cdot(\alpha-2\eps)
\end{align*}
If the buyer's true valuation is $b$ but he deviates and reports $b'$, he is guaranteed a utility of at least $d\cdot (\alpha-2\eps)$. Therefore, a DSIC mechanism must guarantee the buyer at least the same utility from reporting truthfully $b$. Hence:
\begin{align*}
\sum_{i=1}^{d} r_i(b,s^\eps)\geq d\cdot(\alpha-2\eps)
\end{align*}
verifying the induction claim for $d$.
\end{proof}

\begin{proof}[Proof of Theorem \ref{thm:neg}]
Apply Lemma \ref{lem:d=M-1} with $d=M$. 
It implies that every DSIC mechanism with approximation factor at least $\alpha$ must have $\sum_{i=1}^M r_i(b,s^\eps)\geq M\cdot (\alpha-2\eps)$ for every $b\in B^M$. But the sum of all $M$ probabilities must be at most $1$. Therefore, $\alpha \leq 1/M+2\eps$. This is true for every $\eps\in(0,1)$, so $\alpha\leq 1/M$.
\end{proof}

\section{Submodular utilities}
\label{sec:submodular}
In this section, we assume that the utility vectors of both agents satisfy the condition of \emph{decreasing marginal returns}, which is equivalent to submodularity:
\begin{definition}
A utility vector $v_1,\ldots,v_M$ is \emph{submodular} if, for every $i\range{1}{M-1}$:
\begin{align*}
v_{i+1} - v_i ~~ \leq ~~ v_i - v_{i-1}
\end{align*}
\end{definition}
In the case of multi-unit bilateral trade (Example \ref{exm:multiunit}), it is clear that if the buyer's utility function $\beta_i$ is submodular, then his utility vector $b_i$ is submodular too. It can be checked that the same is true for the seller: if $\sigma_i$ is submodular then $s_i$ is submodular too. We now present a randomized mechanism for submodular agents.

\begin{example}
The \emph{Decreasing Random (DR) mechanism} chooses each option $i\in \{1,\ldots,M\}$ with probability:
\begin{align*}
r_i := {1\over i\cdot H_M}
\end{align*}
where $H_M$ is the $M$-th harmonic number:
\begin{align*}
H_M := \sum_{j=1}^M {1\over j} 
~~\in~~ 
[\ln{M}+{1\over M},~ \ln{M}+1]
\end{align*}
It then lets the agents decide whether they want option $i$ or nothing. The agents will cooperate iff $b_i\geq 0,s_i\geq 0$. Therefore, the expected gain is:
\begin{align*}
G_{DR}(b,s)
~~
=
~~
\sum_{i\in\feas(b,s)} {b_i+s_i \over i\cdot H_M}
\end{align*}
Suppose both the buyer and the seller have submodular utilities. 
Let $v_i := b_i + s_i$. Then the vector $v_i$ is submodular too, and this implies that for every $i\range{1}{M}$:
\begin{align*}
{v_{i+1}\over {i+1}} ~~ \leq ~~ {v_{i}\over {i}}
\end{align*}
Therefore, for every $k\range{1}{M}$:
\begin{align*}
G_{DR}(b,s)
~~
&&\geq
~~
\sum_{i\range{1}{k}} {v_i \over i\cdot H_M}
\\
&&\geq
~~
\sum_{i\range{1}{k}} {v_k \over k\cdot H_M}
\\
&&=
~~
k {v_k \over k\cdot H_M}
=
~~
{v_k \over H_M}
\end{align*}
This is true in particular for the option $k$ in which OPT is attained. Therefore, $G_{DR}(b,s) \geq OPT / H_M$. 

This inequality is tight. For example, we might have, for all $i$, $b_i=i$ and $s_i=0$ (both are submodular). Then, the optimal gain is $M$ while $G_{DR}(b,s) = {M\over H_M}$. Therefore, $C_{DR} = {1 / H_M}$.
\end{example}

Our next theorem shows that, even when the agents' utilities are known to be submodular, no DSIC mechanism can be better than the decreasing-random mechanism.
\begin{theorem}
\label{thm:submodular}
When there are $M$ options, every DSIC mechanism that works for arbitrary submodular valuations 
has an expected competitive ratio of at most $1/H_M$.
\end{theorem}
We will prove the theorem by applying Lemma \ref{lem:general} to utility vectors from a special family.
For every $d\range{1}{M}$, let $B_*^{d}$ be the family of submodular vectors $b$ with:
\begin{align*}
b_1 \geq \cdots \geq b_d > 0 \geq b_{d+1} \geq \cdots \geq b_M
\end{align*}

Let $s^*$ be a vector with: $s^*_i = i$ for all $i$ (note that $s^*$ is submodular too).

\begin{lemma}
\label{lem:submodular}
Let $r$ be a DSIC mechanism with competitive ratio $\alpha$. Then, for every $d$ and for every $b\in B_*^d$:
\begin{align*}
\sum_{i=1}^d
r_i(b, s^*)\cdot b_i
\geq
\alpha\cdot \sum_{i=1}^d {b_i\over i}
\end{align*}
\end{lemma}
\begin{proof}
The proof is by induction on $d$.

For $d=1$, we have $b_1=1$ and $b_{i\geq 2}\leq 0$. Applying Lemma \ref{lem:general} to the buyer gives $r_1(b,s^*)\cdot b_1 \geq \alpha\cdot b_1$.

We now assume the claim is true for $d-1$ and prove it is true for $d$. Let $b$ be any vector in $B_*^d$.
Consider the alternative vector $b'$, with:
\begin{align*}
b_i' = 
\begin{cases}
b_i - {i \over d}\cdot b_d & i\range{1}{d}
\\
-1/\eps^M  & i\range{d+1}{M}
\end{cases}
\end{align*}
where $\eps$ is a very small number. By arguments similar to Lemma \ref{lem:d=M-1}, when $\eps$ is sufficiently small, any mechanism with positive competitive ratio must select any option $i\range{d+1}{M}$ with negligible probability. 
So, applying Lemma \ref{lem:general} to the seller gives:
\begin{align}
&&
\sum_{i=1}^d r_i(b',s^*)\cdot i ~~ \geq ~~ s^*_d\cdot \alpha ~~&=~~ \alpha\cdot d
\notag
\\
\label{eq:seller}
\implies &&
\sum_{i=1}^d r_i(b',s^*)\cdot {i\over d}\cdot b_d ~~ &\geq \alpha\cdot b_d
\end{align}

Since $b$ is submodular,
and the number we subtract from every $b_i$ grows with $i$, 
$b'$ is submodular too.
Moreover,
\begin{align*}
b'_1 > b'_2 > \cdots b'_d = 0 > b'_{d+1} > \cdots > b'_M
\end{align*}
Therefore, $b'\in B^{d-1}$ and we can apply the induction assumption:
\begin{align}
\label{eq:buyer}
\sum_{i=1}^{d-1} r_i(b',s^*)\cdot b'_i \geq \alpha\cdot \sum_{i=1}^{d-1} {b'_i\over i}
\end{align}
Note that $b'_d=0$, so we can start the summation in the left-hand side from $d$ and it will not change. Summing up inequalities (\ref{eq:seller}) and (\ref{eq:buyer}) yields:
\begin{align}
\notag
&& 
\sum_{i=1}^d r_i(b',s^*)\cdot (b'_i + {i\over d}\cdot b_d)
~~ &\geq 
\alpha\cdot \left(b_d + \sum_{i=1}^{d-1} {b'_i\over i}\right)
\\
\notag
\implies
&& 
\sum_{i=1}^d r_i(b',s^*)\cdot b_i
~~ &\geq 
\alpha\cdot \left(b_d + \sum_{i=1}^{d-1} ({b_i\over i} - {b_d\over d})\right)
\\
\label{eq:sum}
\implies
&& 
\sum_{i=1}^d r_i(b',s^*)\cdot b_i
~~ &\geq 
\alpha\cdot \sum_{i=1}^{d} {b_i\over i}
\end{align}
If the buyer's true valuation is $b$ but he deviates and reports $b'$, his utility is the left-hand side of (\ref{eq:sum}). A DSIC mechanism must guarantee the buyer at least the same utility from reporting truthfully $b$. Hence:
\begin{align*}
\sum_{i=1}^d r_i(b,s^*)\cdot b_i
~~ &\geq 
\alpha\cdot \sum_{i=1}^{d} {b_i\over i}
\end{align*}
verifying the induction claim for $d$.
\end{proof}

\begin{proof}[Proof of Theorem \ref{thm:submodular}]
Consider the utility vector $b$ in which $b_i = 1$ for all $i$. Note that $b\in B_*^M$ so it is eligible to Lemma \ref{lem:submodular}. Hence:
\begin{align*}
\sum_{i=1}^M r_i(b,s^*)
~~ &\geq 
\alpha\cdot \sum_{i=1}^{M} {1\over i} = \alpha\cdot H_M
\end{align*}
But the sum of probabilities is at most 1, so $\alpha\leq 1/H_M$ as claimed.
\end{proof}

\section{Bilateral Trade}
\label{sec:bilateral}
In this section we prove two corollaries of Theorem \ref{thm:neg} for two settings of bilateral trading.

\begin{corollary}
\label{cor:multi-unit}
A seller holds $M$ units of a single good. The seller values each bundle of $i$ units as $\sigma_i$ and a potential buyer values each bundle of $i$ units as $\beta_i$. 
The price-per-unit $p$ is fixed exogenously. Both agents have utility functions quasi-linear in money. 
Then, the competitive ratio of any DSIC mechanism is:
\begin{itemize}
\item  at most $1/M$ with general valuations, and:
\item  at most $1/H_M$ with submodular valuations.
\end{itemize}
\end{corollary}
\begin{proof}
Apply Theorems \ref{thm:neg} and \ref{thm:submodular} with $b_i = \beta_i - p\cdot i$ and $s_i = p\cdot i + \sigma_{M-i} - \sigma_{M}$.
\end{proof}

\begin{remark}
Obviously there are special cases in which Corollary \ref{cor:multi-unit} does not hold. 
For example, if it is public knowledge that the utilities of both the buyer and the seller are additive (i.e, $\beta_i = i\cdot \beta_1$ and $\sigma_i = i\cdot \sigma_1$ for all $i$), then the following mechanism is DSIC: ask the agents to report $\beta_1$ and $\sigma_1$; trade $M$ units iff $\beta_1 \geq p \geq \sigma_1$.
\end{remark}

\begin{corollary}
\label{cor:multi-type}
A seller holds a set $\mathbb{M}$ containing
$M$ different goods, a single unit of each good.
The seller has an arbitrary valuation function over bundles of goods: $\sigma_i(X)$ is the seller's value for bundle $X$.
The buyer has a unit-demand valuation function: he values each item $i$ as $\beta_i$ and is interested in buying at most one item.
Then, the competitive ratio of any DSIC mechanism is at most $1/M$.
\end{corollary}
\begin{proof}
Apply Theorem \ref{thm:neg} with $b_i = \beta_i - p_i$ and $s_i = p_i + \sigma(\mathbb{M}\{i\}) - \sigma({\mathbb{M}})$.
\end{proof}
\begin{remark}
Corollary \ref{cor:multi-type} holds even if it is public knowledge that the seller has an additive valuation function. However, it does not hold if both the seller and the buyer are additive. In this case, it is possible to trade in each item-type separately using a simple mechanism: trade iff $\beta_i\geq p_i\geq \sigma_i$.
\end{remark}

\section{Conclusion}
We presented impossibility theorems on impossibility of truthful cooperation in general. 
Applying these theorems to Examples \ref{exm:multiunit} and \ref{exm:unitdemand} yields impossibility results on bilateral trade with fixed prices.

Our general results about impossibility of cooperation can potentially be extended in two ways. 

1. We assumed that the mechanism is not allowed to charge money --- so it must be strongly-budget-balanced. What happens if we allow the mechanism to charge or pay money, but require that it remains weakly-budget-balanced (no deficit)?

2. We required the mechanism to be dominant-strategy IC. What happens if we allow it to be IC only in Bayesian-Nash equilibrium, like the original theorem of \citet{myerson1983efficient}?

The applications to bilateral trade can also be extended. Suppose there are $M$ different kinds of goods and the buyer is not unit-demand? If the valuations are arbitrary, then our general impossibility implies that the competitive ratio can be at most $1/2^M$. What is the upper bound when the agents' valuations are sub-additive? sub-modular? gross-substitute? 

\bibliographystyle{apalike}
\bibliography{erel}

\end{document}